\documentclass[conference]{IEEEtran}
\usepackage{amsfonts}
\usepackage{dsfont}
\usepackage{amssymb}
\usepackage{graphicx}
\usepackage{subfigure}
\usepackage{enumerate}
\usepackage{amsmath}
\usepackage{color}
\usepackage{amsthm}
\usepackage{amsmath}
\usepackage{algorithm}
\usepackage{algpseudocode}
\usepackage{breakurl}
\usepackage{cite}
\usepackage{epstopdf}

\newcommand{\bp}{\begin{proof} \small }
\newcommand{\ep}{\end{proof} \normalsize}
\newcommand{\epx}{\end{proof} \small}
\newcommand{\bpa}{\begin{proofappx} \footnotesize }
\newcommand{\epa}{\end{proofappx} \small }
\newtheorem{theorem}{Theorem}

\newtheorem{lemma}{Lemma}

\newtheorem*{theorem*}{Theorem}
\newtheorem*{proposition*}{Proposition}
\newtheorem*{corollary*}{Corollary}
\newtheorem*{lemma*}{Lemma}
\newtheorem*{assumption*}{Assumption}
\newtheorem*{definition*}{Definition}
\newtheorem*{claim*}{Claim}

\newcommand{\be}{\begin{equation}}
\newcommand{\ee}{\end{equation}}
\newcommand{\bs}{\begin{subequations}}
\newcommand{\es}{\end{subequations}}
\newcommand{\bq}{\begin{eqnarray}}
\newcommand{\eq}{\end{eqnarray}}
\newcommand{\bqn}{\begin{eqnarray*}}
\newcommand{\eqn}{\end{eqnarray*}}

\newcommand{\ba}{\left[ \begin{array}}
\newcommand{\ea}{\\ \end{array} \right]}
\newcommand{\ben}{\begin{enumerate}}
\newcommand{\een}{\end{enumerate}}

\def\real{{\mathchoice%
{\hbox{\rm\setbox1=\hbox{I}\copy1\kern-.45\wd1 R}}
{\hbox{\rm\setbox1=\hbox{I}\copy1\kern-.45\wd1 R}}
{\hbox{\scriptsize\rm\setbox1=\hbox{I}\copy1\kern-.45\wd1 R}}
{\hbox{\scriptsize\rm\setbox1=\hbox{I}\copy1\kern-.45\wd1 R}}}}

\def\Zint{{\mathchoice{\setbox1=\hbox{\sf Z}\copy1\kern-.75\wd1\box1}
{\setbox1=\hbox{\sf Z}\copy1\kern-.75\wd1\box1}
{\setbox1=\hbox{\scriptsize\sf Z}\copy1\kern-.75\wd1\box1}
{\setbox1=\hbox{\scriptsize\sf Z}\copy1\kern-.75\wd1\box1}}}
\newcommand{\complex}{ \hbox{\rm C\kern-0.45em\rule[.07em]{.02em}{.58em}%
\kern 0.43em}}

\begin{document}
	%
	\title{Learning-Based Task Offloading for Vehicular Cloud Computing Systems}


	\author{\IEEEauthorblockN{Yuxuan Sun$^*$, Xueying Guo$^\dagger$, Sheng Zhou$^*$, Zhiyuan Jiang$^*$, Xin Liu$^\dagger$, Zhisheng Niu$^*$}
		\IEEEauthorblockA{$^*$Tsinghua National Laboratory for Information Science and Technology\\
			Department of Electronic Engineering, Tsinghua University, Beijing, China\\
            Email: sunyx15@mails.tsinghua.edu.cn, \{sheng.zhou, zhiyuan, niuzhs\}@tsinghua.edu.cn\\
			$^\dagger$Department of Computer Science, University of California, Davis, CA, USA\\
			Email: guoxueying@outlook.com, xinliu@ucdavis.edu}}

	\maketitle

	\begin{abstract}
		Vehicular cloud computing (VCC) is proposed to effectively utilize and share the computing and storage resources on vehicles. However, due to the mobility of vehicles, the network topology, the wireless channel states and the available computing resources vary rapidly and are difficult to predict. 
		In this work, we develop a learning-based task offloading framework using the multi-armed bandit (MAB) theory, which enables vehicles to learn the potential task offloading performance of its neighboring vehicles with excessive computing resources, namely service vehicles (SeVs),  and minimizes the average offloading delay. 
		We propose an adaptive volatile upper confidence bound (AVUCB) algorithm and augment it with load-awareness and occurrence-awareness, by redesigning the utility function of the classic MAB algorithms. The proposed AVUCB algorithm can effectively adapt to the dynamic vehicular environment, balance the tradeoff between exploration and exploitation in the learning process, and converge fast to the optimal SeV with theoretical performance guarantee.  		
		Simulations under both synthetic scenario and a realistic highway scenario are carried out, showing that the proposed algorithm achieves close-to-optimal delay performance.
	\end{abstract}
	

	%
	\IEEEpeerreviewmaketitle
	\section{Introduction}
	Mobile devices are expected to support a vast variety of mobile applications. Many of them require a large amount of computation in a very short time, which cannot be satisfied by the devices themselves due to the limited processing power and battery capacity. 	
	Mobile cloud computing (MCC) is thus proposed \cite{dinh2013a}, enabling mobile devices to offload computation tasks to powerful remote cloud servers through the Internet. However, centralized deployments of clouds introduce long latency for data transmission, which cannot meet the requirements of the emerging delay-sensitive applications, such as augmented/virtual reality, connected vehicles, Internet of things, etc.
	By deploying computing and storage resources at the network edge, mobile edge computing (MEC) can provide low latency computing services \cite{hu2015mobile}. A major challenge in the MEC system is to perform \emph{task offloading}, i.e., when and where to offload the computation tasks, and how to allocate radio and computing resources, which have been widely investigated \cite{mao2017mobile, you2016energy, zhao2017task}. 
	Ad hoc cloudlet has been proposed in \cite{chen2015on} to make use of the computing resources of mobile devices through device-to-device communications \cite{wang2016pairing}.

	Vehicles have huge potential to enhance edge intelligence.
	The global number of connected vehicles is increasing rapidly, and will achieve to around 250 million by 2020 \cite{velosa2014predicts}. Meanwhile, vehicles are equipped with increasing amount of computing and storage resources \cite{abdel2015vehicle}. 
	In order to improve the utilization of vehicle resources, the concept of vehicular cloud computing (VCC) is proposed \cite{surveyvcc}, in which vehicles can serve as vehicular cloud (VC) servers by sharing their surplus computing resources, and users such as other vehicles and pedestrians can offload computation tasks to them. In this case, the vehicles providing services are called service vehicles (SeVs) and the vehicles that offload their tasks are called task vehicles (TaVs).
	Existing architectures include software-defined VC \cite{choosoftware} and VANET-Cloud \cite{bitam2015vanet}.

	Compared with the MEC system, the highly dynamic vehicular environments bring more uncertainties to the VCC system. First, the topology of vehiclar networks and the wireless channel states vary rapidly over time due to the mobility of vehicles. Second, the computing resources of SeVs are heterogeneous and fluctuate over time. These factors are typically difficult to predict, but significantly affect the delay performance of computation tasks. Furthermore, each TaV may be surrounded by multiple candidate SeVs since the density of SeVs can be much higher than that of MEC servers. It is not easy to estimate the performance of different SeVs for task offloading.
	
	There are existing papers investigating task offloading algorithms in the VCC system. In \cite{zheng2015smdp}, the VC and remote cloud layer are jointly considered. A centralized task offloading algorithm is proposed to minimize the average system cost related to delay, energy consumption and resource occupation. However, the centralized control requires large signaling overheads for vehicular states update, and the proposed algorithm has high complexity. A distributed task offloading algorithm is proposed in \cite{fengave} based on ant colony optimization, which is of much lower complexity. However, it still requires exchanges of vehicular states.
	To further overcome the uncertainties in the VCC system and improve service reliability, replicated task offloading is proposed in \cite{jiang2017}, in which task replicas are assigned to multiple SeVs at the same time.
	
	In this work, we focus on the task offloading problem of TaVs in the VCC system, and propose a learning-based task offloading algorithm to minimize the average offloading delay.
	Our main contributions are summarized as follows:
	
	1) We propose a learning-based distributed task offloading framework based on the multi-armed bandit (MAB) theory \cite{auer2002finite}, which enables TaVs to learn the performance of SeVs and to make task offloading decisions individually, in order to obtain low delay without exchanging vehicular states.
	
	2) We propose an adaptive volatile upper confidence bound (AVUCB) algorithm by redesigning the utility function in the classic MAB algorithms, making it adapt to the time-varying load and action space. Both load-awareness and occurrence-awareness are augmented to the learning algorithm, so that AVUCB can effectively cope with the highly dynamic vehicular environment and balance the so-called exploration and exploitation tradeoff in the learning process. We also prove that the performance loss can be upper bounded.
	
	3) Simulations are carried out under both synthetic scenario and a realistic highway scenario, showing that the proposed algorithm can achieve close-to-optimal delay performance.
	
	The rest of this paper is organized as follows. The system model and problem formulation is described in Section \ref{sys}. The AVUCB algorithm is then proposed in Section \ref{algo} and the performance is analyzed in Section \ref{per}. Simulation results are provided in Section \ref{sim}, and finally comes the conclusion in Section \ref{con}.

	\section{System Model and Problem Formulation} \label{sys}
	
	\subsection{System Overview}
	
	\begin{figure}  [!htb]
		\centering
		\includegraphics[width=0.48\textwidth]{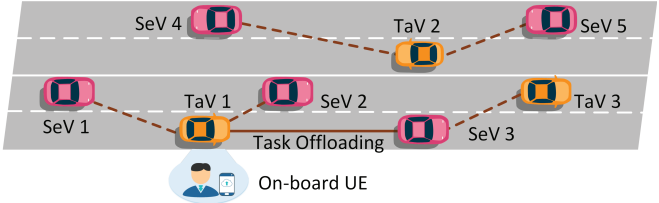}\\
		\caption{An illustration of task offloading in the VCC system. }\label{system}
	\end{figure}
	
    We consider a discrete-time VCC system in which moving vehicles are classified into two categories: SeVs and TaVs.
	SeVs are employed as vehicular cloud servers to provide computing services, while the on-board user equipments (UEs) of TaVs, such as smart phones and laptops, generate computation tasks that need to be offloaded for cloud execution. 
	Note that the role of SeV or TaV is not fixed for each vehicle, which depends on whether the computing resources are sufficient and shareable. 
	For each TaV, the surrounding SeVs in the same moving direction within its communication range $C_r$ are considered as candidate servers. 
	Here the moving direction, together with vehicle ID, speed and location of each candidate SeV can be known to each TaV, provided by vehicular communication protocols such as beacons of dedicated short-range communication (DSRC) standards \cite{kenney2011dsrc}.
	
	Each computation task is offloaded to one of the candidate SeVs and processed on this single SeV according to the task offloading algorithm, without further offloaded to other SeVs, MEC servers or the remote cloud.
	An exemplary VCC system is shown in  Fig. \ref{system}, where TaV 1 discovers 3 candidate SeVs (SeV 1-3) in its neighborhood, and the current task is offloaded to and executed on SeV 3.
	
	In this work, task offloading decisions are made in a distributed manner, i.e., each TaV makes its own task offloading decisions, in order to avoid large signaling exchange overhead. We focus on a representative TaV that moves on the road for a total of $T$ time periods. In time period $t$, the TaV generates a computation task, selects an SeV $n \in \mathcal{N}(t)$, offloads the task and then receives the computing result. Here, $\mathcal{N}(t)$ is denoted as the candidate SeV set that can provide computing services to the TaV in time period $t$.
	We assume that $\mathcal{N}(t) \neq \emptyset$ for $\forall t$, otherwise the TaV can offload tasks to the MEC server or the remote cloud.
	Note that $\mathcal{N}(t)$ changes across time since vehicles are moving.
	
	\subsection{Computation Task Offloading}
	
	The computation task generated in time period $t$ is described by three parameters: input data size $x_t$ (in bits) that needs to be transmitted from TaV to SeV $n$, output data size $y_t$ (in bits) that is fed back from SeV $n$ to TaV, and the computation intensity $w_t$ (in CPU cycles per bit) which is the required CPU cycles to compute one bit input data. Then the total required CPU cycles of the task in time period $t$ is $x_t w_t$ \cite{mao2017mobile}.
	
	For each candidate SeV $n$, the maximum computation capability is denoted by $F_n$ (in CPU cycles per second), which is the maximum available CPU speed of its on-board server. Multiple computation tasks can be processed simultaneously using processor sharing, and the allocated computation capability to the considered TaV in time period $t$ is denoted by $f_{t,n}$. Then the computation delay of SeV $n\in \mathcal{N}(t)$ is
	\begin{align}
	d_c(t,n) = \frac{x_t w_t}{f_{t,n}}.
	\end{align}
	However, in the real system,  $f_{t,n}$ may be unknown to the TaV in advance (this will be discussed in details in Section \ref{pro}).
	
	At time period $t$, the uplink transmission rate between the TaV and each candidate SeV $n\in \mathcal{N}(t)$ is denoted by $r^{(u)}_{t,n}$, which mainly depends on the uplink channel state $h^{(u)}_{t,n}$ between the TaV and SeV $n$, and the interference power $I^{(u)}_{t,n}$ at SeV $n$. Given the channel bandwidth $W$, the transmission power $P$ of the TaV and the noise power $\sigma^2$, the uplink transmission rate can be written as
	\begin{align}
	r^{(u)}_{t,n} = W\log_2\left(1 + \frac{Ph^{(u)}_{t,n}}{\sigma^2 + I^{(u)}_{t,n}}\right).
	\end{align}
	Similarly, the downlink transmission rate is given by
	\begin{align}
	r^{(d)}_{t,n} = W\log_2\left(1 + \frac{Ph^{(d)}_{t,n}}{\sigma^2 + I^{(d)}_{t}}\right),
	\end{align}
	where $h^{(d)}_{t,n}$ is the downlink channel state between SeV $n$ and the TaV, and  $I^{(d)}_{t}$ is the interference at the TaV.

	Therefore, the total transmission delay for uploading the task to SeV $n$ and receiving the result feedback is
	\begin{align}  
	d_t(t,n) = \frac{x_t }{r^{(u)}_{t,n}}+\frac{y_t }{r^{(d)}_{t,n}}.
	\end{align}
	Still, both $r^{(u)}_{t,n}$ and $r^{(d)}_{t,n}$ are unknown to the TaV in advance.
	
	Finally, the sum offloading delay to SeV $n$ in time period $t$ is the computation delay plus the transmission delay
	\begin{align}  
	d(t,n) = d_c(t,n) +d_t(t,n) .
	\end{align}	 
	
	\subsection{Problem Formulation} \label{pro}
	The TaV makes task offloading decisions about which SeV should serve each computation task, in order to minimize the average offloading delay. The problem is formulated as
    \begin{align} \label{obj}
	\textbf{P1:}~\min_{a_1,...,a_T} \frac{1}{T}\sum_{t=1}^{T}d(t,a_t),  
	\end{align}	 
	where $a_t\in\mathcal{N}(t)$ is the index of the selected SeV for task offloading in time period $t$.

	If the TaV knows the exact computation capability $f_{t,n}$, uplink and downlink transmission rates $r^{(u)}_{t,n}$, $r^{(d)}_{t,n}$ of all candidate SeVs before offloading the task in time period $t$, it only needs to calculate the sum delay $d(t,n)$ for $ \forall n \in \mathcal{N}(t)$, and $a_t=\arg \min_n d(t,n)$.
	However, in real systems, the wireless channel state and the interference change rapidly due to the movements of vehicles, and the computing resource of each SeV is shared by multiple tasks. Thus the transmission rates $r^{(u)}_{t,n}$, $r^{(d)}_{t,n}$ and the computation capability $f_{t,n}$ are fast varying across time, which are not easy to predict. On the other hand, if each TaV requests $r^{(u)}_{t,n}$, $r^{(d)}_{t,n}$ and $f_{t,n}$ of all candidate SeVs in each time period, the signaling overhead will be very high. 
	
	Without the pre-knowledge of the transmission rates $r^{(u)}_{t,n}$, $r^{(d)}_{t,n}$ and computation capability $f_{t,n}$ of each candidate SeV $n$, the TaV does not know which SeV performs the best when making the current offloading decision.
	Therefore, we will design learning-based task offloading algorithm in the following section, in which the TaV learns the delay performance of candidate SeVs based only on the historical delay observations. That is, the offloading decision $a_{t}$ at time $t$ is based on the observed delay sequence $d(1,a_1), d(2,a_2), ...,d(t-1,a_{t-1})$, but not the exact value of $r^{(u)}_{t,n}$, $r^{(d)}_{t,n}$ and $f_{t,n}$ of any candidate SeV $n$ in the current time period $t$.

	\section{Learning-Based Task Offloading Algorithm}\label{algo}
	In this section, we develop learning-based task offloading algorithm which enables the TaV to learn the delay performance of candidate SeVs and minimizes the average offloading delay.
	
	We assume that tasks are of diverse input data size, but the ratio of output data size and input data size, as well as the computation intensity are identical. In fact, this is a valid assumption when the offloaded tasks are of the same kind. Let $y_t/x_t=\alpha_0$ and $w_t=\omega_0$ for $\forall t$. Define the bit offloading delay as
	\begin{align}
	u(t,n)=\frac{1 }{r^{(u)}_{t,n}}+\frac{\alpha_0}{r^{(d)}_{t,n}}+ \frac{\omega_0}{f_{t,n}},
	\end{align}	
	 which is the sum delay of offloading one bit input data to SeV $n$ in time period $t$, reflecting the comprehensive service capability of each candidate SeV. Therefore, the sum delay 
	\begin{align}  
	d(t,n) = x_tu(t,n).
	\end{align}
	When making the offloading decision of the current task at time $t$, TaV knows the input data size $x_t$. But for $\forall n\in \mathcal{N}(t)$, the exact value of $u(t,n)$ and its distribution are not known to the TaV in prior, which need to learn.
	
	
	Our task offloading problem can be formulated as an MAB problem to solve.
	To be specific, the TaV is the player and each candidate SeV corresponds to an action with unknown distribution of loss. 
	The player makes sequential decisions on which action should be taken to minimize the average loss.
	The main challenge of the classic MAB problem is to balance the exploration and exploitation tradeoff: explore different actions to learn good estimates of each distribution, while at the same time select the empirically best actions as many as possible. The problem has been widely studied and many algorithms have been proposed with strong performance guarantee, such as the upper confidence bound (UCB) based UCB1 and UCB2 algorithms \cite{auer2002finite}.
	The MAB framework has already been applied in the wireless networks to help learn the unknown environments, such as solving channel access problems \cite{chen2011opp} and mobility management issues \cite{sun2017emm}.
	
	Although our problem is similar to the classic MAB problem, we still face two new challenges. 
	First, the candidate SeV set $\mathcal{N}(t)$ changes across time due to the relative movements of vehicles, rather than the fixed number of actions in the classic MAB problem. The SeVs may appear and disappear in the communication range of the TaV unexpectedly, causing a volatile action space. Existing solutions cannot exploit the empirical information of the remaining SeVs efficiently.
	Second, the performance loss in each time period is of equal weight in the MAB problem. However, in our model, the input data size $x_t$ of each task brings a weighting factor on the offloading delay. Intuitively, the task offloading algorithm should explore more when $x_t$ is low, and exploit more when $x_t$ is high, so that the cost of exploration can be reduced. 
	
	To overcome the aforementioned two challenges, we propose an Adaptive Volatile UCB (AVUCB) algorithm for task offloading, as shown in Algorithm 1. Parameter $\beta$ is a constant factor, $k_{t,n}$ is the number of tasks that have been offloaded to SeV $n$ up till time $t$, and $t_n$ records the occurrence time of each SeV $n$. Parameter $\tilde{x}_t$ is the normalized input data size within $[0,1]$, which is denoted as
	\begin{align} \label{normailized_x}
		\tilde{x}_t=\max \left\{ \min \left( \frac{x_t-x^-}{x^+-x^-},1\right), 0\right\},
	\end{align}
	where $x^+$ and $x^-$ are the upper and lower thresholds for normalizing $x_t$.

	\begin{algorithm}
		\caption{AVUCB Algorithm for Task Offloading}
		\begin{algorithmic}[1]
			\State \textbf{Input}: $\alpha_0$, $\omega_0$ and $\beta$.
			\For {$t=1,...,T$}
			\If { Any SeV $n \in \mathcal{N}(t)$ has not connected to TaV}
			\State Connect to SeV $n$ once.
			\State Update $\bar u_{t,n}=d(t,n)/x_t$, $k_{t,n}=1$, $t_n=t$.
			\Else
			\State Observe $x_t$.
			\State Calculate the utility function of each candidate SeV $n \in \mathcal{N}(t)$:
			\begin{align} \label{ada_utility}
				\hat{u}_{t,n}=\bar{u}_{t-1,n}-\sqrt{\frac{\beta(1-\tilde{x}_t)\ln (t-t_n)}{k_{t-1,n}}}.
			\end{align}		
			\State Offload the task to SeV:
			\begin{align} \label{ada_obj}
				a_t=\arg\min_{n \in \mathcal{N}(t)} \hat{u}_{t,n}.
			\end{align} 
			\State Observe delay $d(t,a_t)$.
			\State Update $\bar{u}_{t,a_t}\leftarrow \frac{ \bar{u}_{t-1,a_t}k_{t-1,a_t}+d(t,a_t)/x_t}{k_{t-1,a_t}+1} $.
			\State Update $k_{t,a_t}\leftarrow k_{t-1,a_t}+1$.
			\EndIf
			\EndFor		
		\end{algorithmic}
	\end{algorithm}
	
	Our proposed AVUCB algorithm can effectively balance the exploration and exploitation under the variation of candidate SeV set and input data size, inspired by the volatile MAB \cite{bnaya2013social} and opportunistic MAB \cite{wu2017adaptive} frameworks. In Algorithm 1, Lines 3-5 are the initialization phase, in which the TaV will connect to the newly appeared SeV once. 
	Lines 7-12 represent the continuous learning phase. The utility function defined in \eqref{ada_utility} is the sum of empirical delay performance $\bar u_{t,n}$ and a padding function (the latter term). 
	Compared with existing UCB algorithms, the padding function is redesigned by jointly taking into account the occurrence time $t_n$ of SeV and the input data size $x_t$ (the load), thus it can dynamically adjusts the weight of exploration and bring the \emph{occurrence-awareness} and \emph{load-awareness} to the algorithm.
	Task offloading decision is then made in Line 9, which is a minimum seeking problem with computational complexity $O(|\mathcal{N}(t)|)$, where $|\mathcal{N}(t)|$ is the number of candidate SeVs in time period $t$.
	
	
%
%

	\section{Performance Analysis}\label{per}
	In this section, we present the performance analysis of the proposed AVUCB algorithm.
	We first define an epoch as the duration in which the candidate SeV set remains the same. Let $B$ denote the total number of epochs during $T$ time periods, $\mathcal{N}_b$ the candidate SeV set in the $b$th epoch, and $t_b$, $t'_b$ the beginning and ending time period of the $b$th epoch with $t'_B=T$.
	We assume that the bit offloading delay $u(t,n)$ of each candidate SeV is i.i.d. over time and independent of each other, with expectation $\mathbb{E}_t[u(t,n)]=\mu_n$.
	We will prove later through simulations that without this assumption, AVUCB still works well.
	In each epoch, let $\mu_b^*=\min_{n\in\mathcal{N}_b}\mu_n$ and $a_b^*=\arg\min_{n\in\mathcal{N}_b}\mu_n$.

	Define the total \emph{learning regret} as
	\begin{align}
	R_T = \sum_{b=1}^{B} \mathbb{E}\left[\sum_{t=t_b}^{t'_b}x_t\left(u(t,n)-\mu_b^*\right)\right],	
	\end{align}
	which is the expected loss of delay performance due to lacking the service capability information of the candidate SeVs. In the following subsections, we try to upper bound the learning regret of AVUCB algorithm.
	
	
	\subsection{Regret Analysis under Identical Load}
	Compared to the existing UCB algorithms \cite{auer2002finite}, the AVUCB algorithm adds the occurrence time $t_n$ and normalized input data size $\tilde{x}_t$ to the padding function. In this subsection, we first investigate the impact of the occurrence time, by assuming that the load is not time varying, i.e., tasks are of identical input data size with $x_t=x_0$ and $\tilde{x}_t=0$ for $\forall t$.
	Thus the padding function in \eqref{ada_utility} can be simplified as $\sqrt{\frac{\beta\ln (t-t_n)}{k_{t-1,n}}}$, and the learning regret 
	$R_T =x_0 \sum_{b=1}^{B} \mathbb{E}\left[\sum_{t=t_b}^{t'_b}(u(t,n)-\mu_b^*)\right]$.	
	
	Let $u_m=\sup_{t,n} u(t,n) $, and $\delta_{n,b}=(\mu_n-\mu_b^*)/u_m$. We first provide the upper bound of the learning regret within each epoch, as shown in Lemma \ref{vucb_epoch}.
	
	\begin{lemma} \label{vucb_epoch}
		Let $\beta=2u^2_m$, in the $b$th epoch, the learning regret of AVUCB with identical load has an upper bound as follows:
		\begin{align}
			R_{b}\leq x_0u_m \left[\sum_{n\neq a_b^*}\frac{8\ln (t'_b-t_n)}{\delta_{n,b}} + \left(1+ \frac{\pi^2}{3}\right)\sum_{n\neq a_b^*}\delta_{n,b} \right].
		\end{align}
	\end{lemma}
	\begin{proof}
		See Appendix \ref{a1}.
	\end{proof}
	
	Then we can upper bound the learning regret over $T$ time periods in the following theorem.
			\begin{theorem} \label{vucb_T}
				Let $\beta=2u^2_m$, the total learning regret $R_T$ of AVUCB with identical load has an upper bound as follows:
				\begin{align}
				R_{T}\leq x_0u_m\sum_{b=1}^{B} \left[\sum_{n\neq a_b^*}\frac{8\ln T}{\delta_{n,b}} + O(1) \right].
				\end{align}
			\end{theorem}
			\begin{proof}
			See Appendix \ref{a3}.
			\end{proof}

	Theorem 1 implies that when tasks are of equal input data size, the proposed AVUCB algorithm can provide a bounded performance loss, compared to the optimal solution in which the TaV knows in prior which SeV performs the best. Specifically, the performance loss grows linearly with $B$ and logarithmically with $T$.
	
	\subsection{Impact of the Load}
	In this subsection, we show the impact of the load on the learning regret, by considering that the input data size $x_t$ is random and continuous. For simplicity, we focus on a single epoch and assume $B=1$. 
	Thus there exists single best SeV with $\mu^*=\min_{n\in\mathcal{N}_1}\mu_n$, $a^*= \arg\min_{n\in\mathcal{N}_1}\mu_n$, and the learning regret is simplified as $R_T =\mathbb{E}\left[\sum_{t=1}^{T}x_t(u(t,n)-\mu^*)\right]$. 
	
	Recall that the normalized input data size $\tilde{x}_t$ is defined in \eqref{normailized_x}, in which the upper and lower thresholds $x^+$ and $x^-$ should be carefully selected for the tradeoff between exploration and exploitation.  	
	In the following theorem, $x^+$ and $x^-$ are selected such that $\mathbb{P}\{x_t\leq x^-\}>0$ and $x^+\geq	x^-$. Particularly, when $x^+=	x^-$, let $\tilde{x}_t=0$ if $x_t\leq  x^-$ and $\tilde{x}_t=1$ if $x_t >  x^-$.
	
	The learning regret under random and continuous input data size is shown in Theorem 2.
	
		\begin{theorem} \label{ada}
			Let $\beta=2u^2_m$, with random continuous $x_t$ and $\mathbb{P}\{x_t\leq x^-\}>0$, we have:
			
		(1)	With $x^+\geq	x^-$, the expected number of tasks $k_{T,n}$ offloaded to any SeV $n\neq a^*$ can be upper bounded as 
		\begin{align}
		\mathbb{E}[k_{T,n}] \leq\frac{8\ln T }{\delta^2_{n}} +O(1).   
		\end{align}
		
		(2) With $x^+=	x^-$, the learning regret
			\begin{align}
				R_T \leq u_m\sum_{n\neq a^*}\left[ \frac{8\ln T \mathbb{E}[x_t|x_t\leq x^-]}{\delta_{n}} +O(1)    \right],
			\end{align}
			where $\mathbb{E}[x_t|x_t\leq x^-]$ is the expected $x_t$ under condition $x_t\leq x^-$, $u_m=\sup_{t,n} u(t,n) $, and $\delta_{n}=(\mu_n-\mu^*)/u_m$.
		\end{theorem}
		\begin{proof}
			See Appendix \ref{a2}.
		\end{proof}
		
		Theorem 2 shows that, under single epoch, the learning regret grows logarithmically with $T$. This implies that when the input data size $x_t$ varies across time, our proposed AVUCB algorithm can still effectively balance the exploration and exploitation by adjusting the normalized factor $\tilde{x}_t$, and provide a bounded deviation compared to the optimal solution.  
	
	\section{Simulations} \label{sim}
	In this section, we evaluate the average delay and learning regret of the proposed AVUCB algorithm through simulations. We start from a synthetic scenario, and then simulate a realistic highway scenario.
	
	\subsection{Simulation under Synthetic Scenario} 
	Simulations under synthetic scenario are carried out in MATLAB. We consider one TaV and 5 SeVs appear (and may also disappear) in the duration of $T=1200$ time periods. The communication range $C_r=200 \mathrm{m}$. Within the TaV's communication range, the distance between the TaV and each SeV ranges in $[10, 200]\mathrm{m}$, and changes randomly by $-10\mathrm{m}$ to $10\mathrm{m}$ every time period. According to \cite{abdulla2016vehicle}, the wireless channel state is modeled by an inverse power law $h^{(u)}_{t,n}=h^{(d)}_{t,n}=A_0l^{-2}$, where $l$ is the distance between TaV and SeV and $A_0=-17.8\mathrm{dB}$. Other default parameters are: computation intensity $\omega_0=1000 \mathrm{Cycles/bit}$, transmit power $P=0.1\mathrm{W}$, channel bandwidth $W=10\mathrm{MHz}$, noise power $\sigma^2=10^{-13}\mathrm{W}$ and parameter $\beta$ in \eqref{ada_utility} is $2$.
	
	We first evaluate the effect of occurrence time $t_n$ in the proposed AVUCB algorithm, by assuming that all the tasks are of equal input data size $x_0=0.6 \mathrm{Mbits}$, and thus $\tilde{x}_t=0$ for $\forall t$. The whole duration is divided into $3$ epochs, each having $400$ time periods. The index of candidate SeVs and their maximum computation capability $F_n$ is shown in table \ref{SeV_VUCB}. In epoch 2, there appears two SeVs indexed by $3$ and $4$, while in epoch 3, there appears a new SeV $5$ and SeV $3$ disappears. At each time period, the allocated computation capability $f_{t,n}$ is randomly distributed from $20\%F_n$ to $50\%F_n$.

	     \begin{table}[!htb]
	     	\vspace{-0.1in}
	     	\caption{Candidate SeVs and Maximum Computation Capability}
	     	\label{SeV_VUCB}
	     	\centering
	     	\begin{tabular}{||c||c|c|c|c|c||}
	     		\hline
	     		Index of SeV & 1 & 2 & 3&4&5\\
	     		\hline
	     		$F_n$ (GHz)& 3&4&6&5&2\\
	     		\hline
	     		Epoch 1 &${\surd}$ & ${\surd}$ & -- & --&-- \\
	     		\hline
	     		Epoch 2 &${\surd}$& ${\surd}$ &  ${\surd}$  &${\surd}$ &-- \\
	     		\hline
	     		Epoch 3 & ${\surd}$&${\surd}$& $\times$ & ${\surd}$  &${\surd}$ \\
	     		\hline
	     	\end{tabular}
	     	\vspace{-0.1in}
	     \end{table}

	 	Fig. \ref{Performance_VUCB1} shows the learning regret of the proposed AVUCB algorithm and existing UCB1 algorithm under diverse occurrence time of SeVs and identical load. In the first epoch, the two algorithms perform the same since the occurence time of all SeVs are 1.
	 	From the second epoch, by taking into account of the occurrence time, AVUCB is able to learn the performance of the newly appeared SeVs faster, while effectively making use of the information of the remaining SeVs. Compared to UCB1 algorithm, it can reduce about $50\%$ of the learning regret. 
	 	The average delay performance of each epoch is shown in Fig. \ref{Performance_VUCB2}, in which the average delay of AVUCB converges faster to optimal delay than UCB1.
	 
	\begin{figure}[!t]
		\centering	
		\subfigure[Learning regret]{\label{Performance_VUCB1}			
			\includegraphics[width=0.4\textwidth]{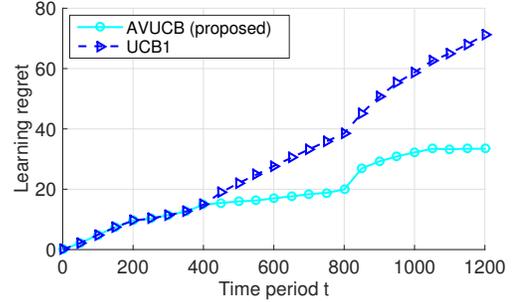}}		
		\subfigure[Average delay]{\label{Performance_VUCB2}	
			\includegraphics[width=0.4\textwidth]{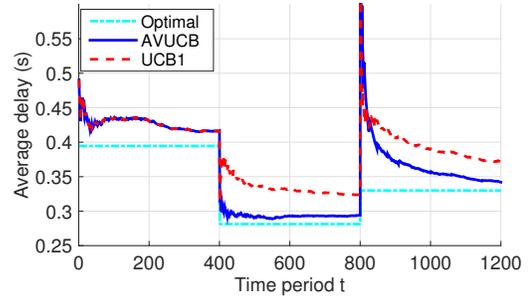}}
		\caption{Performance of AVUCB under diverse SeV occurrence time and identical load.}
		\label{Performance_VUCB}
		\vspace{-3mm}
	\end{figure}
	
	We then focus on a single epoch with $B=1$, and evaluate the effect of the normalized input data size $\tilde{x}_t$ under random load. The candidate SeV set and its maximum computation capability are the same as epoch $2$ in table \ref{SeV_VUCB}. The input data size $x_t$ is uniformly distributed within $[0.2, 1]\mathrm{Mbits}$. The upper and lower thresholds are set to be $x^+=0.8, x^-=0.4$ and $x^+=x^-=0.6$ respectively. As shown in Fig. \ref{Performance_Ada}, under the two sets of thresholds, AVUCB achieves similar learning regret, since both settings can effectively adjust the weight of exploration and exploitation. However, the UCB1 algorithm suffers much higher learning regret without load-awareness, since it may still explore even if the input data size is large.
	
		\begin{figure}[!t]
			\centering	
			\includegraphics[width=0.35\textwidth]{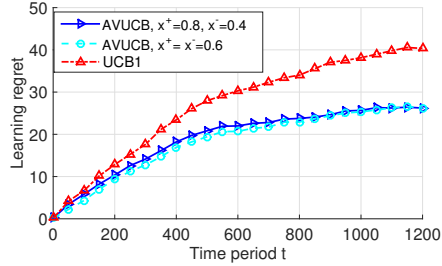}	
			\caption{Performance of AVUCB under random load.}
			\label{Performance_Ada}
			\vspace{-5mm}
		\end{figure}

	\subsection{Simulation under Realistic Highway Scenario} 
	
	In this subsection, we simulate a realistic highway scenario and better emulate the traffic flow using Simulation of Urban MObility (SUMO)\footnote{http://www.sumo.dlr.de/userdoc/SUMO.html}. We use a $12\mathrm{km}$ stretch of two-lane G6 Highway in Beijing with two ramps, obtained from Open Street Map (OSM)\footnote{http://www.openstreetmap.org/}.
	The network consists of $50\%$ SeVs and $50\%$ TaVs. Each vehicle is of equal probability of reaching a maximum speed of $72\mathrm{km/h}$ or $90\mathrm{km/h}$. In each time period, the probability of generating a vehicle from each ramp is $0.1$, and whenever a vehicle approaches a ramp, it leaves the highway with probability $0.5$.
	
	The locations of vehicles are simulated by SUMO, and then we can calculate the distance between each TaV and SeV at each time period. SeVs within its communication range $C_r=200\mathrm{m}$ are considered as candidate SeVs. We then focus on a single TaV, moving through the highway in $T=400$ time periods. The occurrence and departure time of each candidate SeV, and its maximum computation capability $F_n$ are listed in table \ref{SeV_SUMO}. Same as above, the allocated computation capability $f_{t,n}$ is randomly distributed from $20\%F_n$ to $50\%F_n$, and the input data size $x_t$ is uniformly distributed within $[0.2, 1]\mathrm{Mbits}$. Let $x^+=x^-=0.6$.

 \begin{table}[!htb]
		     	\vspace{-0.1in}
		     	\caption{Candidate SeVs and Maximum Computation Capability}
		     	\label{SeV_SUMO}
		     	\centering
		     	\begin{tabular}{||c||c|c|c|c|c||}
		     		\hline
		     		Index of SeV & 1 & 2 & 3&4&5\\
		     		\hline
		     		Occurrence time& 1&1&1&118&320\\
		     		\hline
		     		Departure time &400 &400 &400&400&343 \\
		     		\hline
		     		$F_n$ (GHz)&3& 2 &  2.5 &4.5 & 3.5 \\
		     		\hline				 
		     	\end{tabular}
 \end{table}
				     
	We evaluate the average delay performance of the proposed AVUCB algorithm compared with 4 other algorithms: 
	1) \textbf{UCB1} \cite{auer2002finite}, which considers neither occurrence time nor input data size. 2) \textbf{VUCB1} \cite{bnaya2013social}, which is occurrence-aware but not load-aware. 3) A naive \textbf{Random Policy} in which the TaV randomly select a SeV in each time period. 4) \textbf{Optimal Policy}, in which TaV knows the performance of each candidate SeV in advance and selects the optimal one.

\begin{figure}[!t]
	\centering	
	\includegraphics[width=0.4\textwidth]{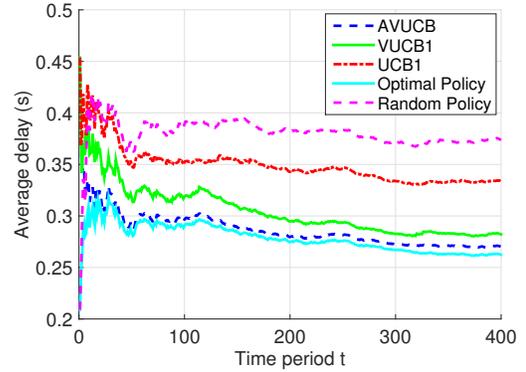}	
	\caption{Average delay performance of AVUCB algorithm under a realistic highway scenario.}
	\label{Performance_SUMO}
	\vspace{-5mm}
\end{figure}

Fig. \ref{Performance_SUMO} shows the average delay of the aforementioned 5 algorithms. Our proposed AVUCB algorithm achieves close-to-optimal delay performance, while outperforms the other algorithms. This is because by introducing the occurrence time of each SeV and the normalized input data size, AVUCB is both occurrence-aware and load-aware, and can effectively balance exploration and exploitation in the learning process.

	\section{Conclusions} \label{con}
	In this work, we have studied the task offloading problem in the VCC system and proposed a learning-based AVUCB algorithm that minimizes the average offloading delay based only on the historical delay observations. The proposed algorithm is of low complexity and easy to implement with low signaling overhead. 
	We have extended the classic MAB algorithms to be both load-aware and occurrence-aware, by taking into account the input data size of tasks and the occurrence time of SeVs in the utility function. 
	Therefore, AVUCB can effectively adapt to the highly dynamic vehicular environment and balance the tradeoff between exploration and exploitation in the learning process with performance guarantees. Simulations under both synthetic scenario and a realistic highway scenario have shown that our proposed algorithm can achieve close-to-optimal delay performance. Future research directions include considering the heterogeneous cloud architecture with VCC, MEC and remote cloud, as well as the cooperation of SeVs.

\section*{Acknowledgment}
This work is sponsored in part by the Nature Science Foundation of China (No. 91638204, No. 61571265, No. 61621091), NSF through grants CNS-1547461, CNS-1718901, CCF-1423542, and Intel Collaborative Research Institute for Mobile Networking and Computing.

\appendices{}


\section{Proof of Lemma 1} \label{a1}

The learning regret in the $b$th epoch can be written as
\begin{align}  \label{a1_1}
R_{b}&=x_0  \mathbb{E}\left[\sum_{t=t_b}^{t'_b}u(t,n)-\mu_b^*\right] \nonumber\\
&=x_0 \mathbb{E}\left[\sum_{n\in \mathcal{N}_b} k_{n,b}u_m\delta_{n,b} \right]=x_0u_m\sum_{n\neq a_b^*} \delta_{n,b} \mathbb{E}[k_{n,b}],
\end{align}
where $k_{n,b}$ is the number of tasks offloaded to SeV $n$ in epoch $b$.
Following the proof of Lemma 1 in \cite{bnaya2013social} and Theorem 1 in \cite{auer2002finite}, when $\beta=2u_m^2$, the expectation of $k_{n,b}$ can be upper bounded by
\begin{align} \label{a1_2}
\mathbb{E}[k_{n,b}] &\leq \frac{8\ln (t'_b-t_n)}{\delta^2_{n,b} }+1+\frac{\pi^2}{3}.
\end{align}
Substituting \eqref{a1_2} into \eqref{a1_1}, we prove Lemma 1:
\begin{align}
&R_{b} =x_0u_m\sum_{n\neq a_b^*} \delta_{n,b} \mathbb{E}[k_{n,b}] \nonumber\\
&\leq x_0u_m \left[\sum_{n\neq a_b^*}\frac{8\ln (t'_b-t_n)}{\delta_{n,b}} + \left(1+ \frac{\pi^2}{3}\right)\sum_{n\neq a_b^*}\delta_{n,b} \right].
\end{align}

\section{Proof of Theorem 1} \label{a3}
According to Lemma 1, the learning regret of the $b$th epoch: 
\begin{align}
R_{b}& \leq x_0u_m \left[\sum_{n\neq a_b^*}\frac{8\ln (t'_b-t_n)}{\delta_{n,b}} + \left(1+ \frac{\pi^2}{3}\right)\sum_{n\neq a_b^*}\delta_{n,b} \right]\nonumber\\
& \leq  x_0u_m \left[\sum_{n\neq a_b^*}\frac{8\ln T}{\delta_{n,b}} +O(1)\right].
\end{align}

By summing over $b=1,2,...,B$, the total learning regret:
\begin{align}
R_{T} =\sum_{b=1}^{B}R_{b}\leq  x_0u_m \sum_{b=1}^{B}\left[\sum_{n\neq a_b^*}\frac{8\ln T}{\delta_{n,b}} +O(1)\right].
\end{align}


\section{Proof of Theorem 2} \label{a2}
When $\beta=2u^2_m$ and $B=1$, the utility function in \eqref{ada_utility} is
\begin{align} 
\hat{u}_{t,n}=\bar{u}_{t-1,n}-u_m\sqrt{\frac{2(1-\tilde{x}_t)\ln t}{k_{t-1,n}}}.
\end{align}	

The offloading decision in \eqref{ada_obj} can be written as:
\begin{align} \label{a2_1}
a_t&=\arg\min_{n \in \mathcal{N}_1} \hat{u}_{t,n} \nonumber\\
&=\arg\min_{n \in \mathcal{N}_1} \left\{  \bar{u}_{t-1,n}-u_m\sqrt{\frac{2(1-\tilde{x}_t)\ln t}{k_{t-1,n}}} \right\}\nonumber\\
&=\arg\min_{n \in \mathcal{N}_1}\left\{ \frac{\bar{u}_{t-1,n}}{u_m}-\sqrt{\frac{2(1-\tilde{x}_t)\ln t}{k_{t-1,n}}} \right\} \nonumber\\
&=\arg\max_{n \in \mathcal{N}_1}\left\{1- \frac{\bar{u}_{t-1,n}}{u_m}+\sqrt{\frac{2(1-\tilde{x}_t)\ln t}{k_{t-1,n}}} \right\}.
\end{align}  

The learning regret can be written as
\begin{align} \label{a2_2}
&~R_T =\mathbb{E}\left[\sum_{t=1}^{T}x_t(u(t,n)-\mu^*)\right]    \nonumber\\
&=u_m\mathbb{E}\left[\sum_{t=1}^{T}x_t\left\{\left(1-\frac{\mu^*}{u_m}\right)-\left(1-\frac{u(t,n)}{u_m}\right)\right\}\right].
\end{align}

Since $1- \frac{\bar{u}_{t-1,n}}{u_m} \in [0,1]$, and $1-\frac{u(t,n)}{u_m} \in [0,1]$, our problem is equivalent to the problem defined in Section \ref{sys} in our previous work \cite{wu2017adaptive}.
By leveraging Lemma 7, Theorem 3 and Appendix C.2 in \cite{wu2017adaptive}, we can get Theorem 2.

	%

\end{document}